\documentclass[11pt]{amsart}

\usepackage[french,english]{babel}

\usepackage{amssymb,amsmath,latexsym}
\usepackage{fullpage}
\usepackage{algorithmicx}
\usepackage{algorithm}
\usepackage{algpseudocode}
\usepackage[round,comma]{natbib}
\usepackage[linktocpage=true,colorlinks=true,citecolor=blue]{hyperref}

\newtheorem{theorem}{Theorem}
\newtheorem{proposition}{Proposition}
\newtheorem{lemma}{Lemma}

\theoremstyle{remark}
\newtheorem{remark}{Remark}

\newcommand{\R}{\mathbb{R}}

\newcommand{\zero}{\boldsymbol{0}}
\renewcommand{\AA}{\boldsymbol{A}}

\newcommand{\ee}{\boldsymbol{e}}

\newcommand{\pp}{\boldsymbol{\pi}}

\newcommand{\xx}{\boldsymbol{x}}
\newcommand{\yy}{\boldsymbol{y}}

\newcommand{\M}{\mathcal{M}}
\newcommand{\X}{\mathcal{X}}

\newcommand{\diag}{\operatorname{diag}}

\newcommand{\supp}{\operatorname{supp}}

\newcommand{\mincost}{\operatorname{mincost}}

\title{Computing solutions of the multiclass network equilibrium problem with affine cost functions}

\author{Fr\'ed\'eric Meunier}
\author{Thomas Pradeau}


\begin{document}

\begin{abstract}
We consider a nonatomic congestion game on a graph, with several classes of players. Each player wants to go from its origin vertex to its destination vertex at the minimum cost and all players of a given class share the same characteristics: cost functions on each arc, and origin-destination pair. Under some mild conditions, it is known that a Nash equilibrium exists, but the computation of an equilibrium in the multiclass case is an open problem for general functions. We consider the specific case where the cost functions are affine. We show that this problem is polynomially solvable when the number of vertices and the number of classes are fixed. In particular, it shows that the parallel-link case with a fixed number of classes is polynomially solvable. On a more practical side, we propose an extension of Lemke's algorithm able to solve this problem.
\end{abstract}

\maketitle

\section{Introduction}

\subsubsection*{Context} Being able to predict the impact of a new infrastructure on the traffic in a transportation network is an old but still important objective for transport planners. In 1952, \cite{Wa52} noted that after some while the traffic arranges itself to form an equilibrium and formalized principles characterizing this equilibrium. With the terminology of game theory, the equilibrium is a Nash equilibrium for a congestion game with nonatomic users. In 1956, \cite{Be56} translated these principles as a mathematical program which turned out to be convex, opening the door to the tools from convex optimization. The currently most commonly used algorithm for such convex programs is probably the Frank-Wolfe algorithm~\citep{FW56}, because of its simplicity and its efficiency, but many other algorithms with excellent behaviors have been proposed, designed, and experimented. 

One of the main assumptions used by Beckmann to derive his program is the fact that all users are equally impacted by the congestion. With the transportation terminology, it means that there is only one {\em class}. In order to improve the prediction of traffic patterns, researchers started in the 70s to study the {\em multiclass} situation where each class has its own way of being impacted by the congestion. Each class models a distinct mode of transportation, such as cars, trucks, or motorbikes. \cite{Da72,Da80} and \cite{Sm79} are probably the first who proposed a mathematical formulation of the equilibrium problem in the multiclass case. However, even if this problem has been the topic of many research works, an efficient algorithm for solving it is not known, except in some special cases~\citep{F77,H88,MM88,MW04}. In particular, there are no general algorithms in the literature for solving the problem when the cost of each arc is in an affine dependence with the flow on it. 

Our purpose is to discuss the existence of such algorithms.

\subsubsection*{Model}

We are given a directed graph $D=(V,A)$ modeling the transportation network. The set of all paths (resp. $s$-$t$ paths) is denoted by $\mathcal{P}$ (resp. $\mathcal{P}_{(s,t)}$). The population of {\em players} is modeled as a bounded real interval $I$ endowed with the Lebesgue measure $\lambda$, the {\em population measure}. The set $I$ is partitioned into a finite number of measurable subsets $(I^k)_{k\in K}$ -- the {\em classes} -- modeling the players with same characteristics: they share a same collection of cost functions $(c_a^k:\mathbb{R}_+\rightarrow\mathbb{R}_+)_{a\in A}$, a same origin $s^k$, and a same destination $t^k$. A player in $I^k$ is said to be of {\em class $k$}. We define $V^k$ (resp. $A^k$) to be the set of vertices (resp. arcs) reachable from $s^k$ in $D$. 

A {\em strategy profile} is a measurable mapping $\sigma:I\rightarrow\mathcal{P}$ such that $\sigma(i)\in\mathcal{P}_{(s^k,t^k)}$ for all $k\in K$ and all $i\in I^k$. We denote by $x_a^k$ the number of class $k$ players $i$ such that $a$ is in $\sigma(i)$: $$x_a^k=\lambda\{i\in I^k:\,a\in\sigma(i)\}.$$ The vector $\xx^k=(x_a^k)_{a\in A^k}$ is an $s^k$-$t^k$ flow of value $\lambda(I^k)$: for each $v\in V^k\setminus\{s^k,t^k\}$, we have $$\sum_{a\in\delta^+(v)}x_a^k=\sum_{a\in\delta^-(v)}x_a^k$$ and $$\sum_{a\in\delta^+(s^k)}x_a^k-\sum_{a\in\delta^-(s^k)}x_a^k=\sum_{a\in\delta^-(t^k)}x_a^k-\sum_{a\in\delta^+(t^k)}x_a^k=\lambda(I^k).$$ The vector $(\xx^k)_{k\in K}$ is thus a multiflow. The total number of players $i$ such that $a$ is in $\sigma(i)$ is $\sum_{k\in K}x_a^k$ and is denoted $x_a$. We denote by $\xx$ the vector $(x_a)_{a\in A}$.

The cost of arc $a$ for a class $k$ player is $c_a^k(x_a)$. For player, the cost of a path $P$ is defined as the sum of the costs of the arcs contained in $P$. Each player wants to select a minimum-cost path.

A strategy profile is a (pure) Nash equilibrium if each path is only chosen by players for whom it is a minimum-cost path. In other words, a strategy profile $\sigma$ is a Nash equilibrium if for each class $k\in K$ and each player $i\in I^k$ we have
$$\sum_{a\in\sigma(i)}c_a^k(x_a)=\min_{P\in\mathcal{P}_{(s^k,t^k)}}\sum_{a\in P}c_a^k(x_a).$$

This game enters in the category of {\em nonatomic congestion games with player-specific cost functions}, see \cite{Mi96}. Under mild conditions on the cost functions, a Nash equilibrium is always known to exist. The original proof of the existence of an equilibrium was made by \cite{Sc70} and uses a fixed point theorem. The proof of this result is also made in \citet{Mi00} or can be deduced from more general results \citep{Ra92}.

The problem of finding a Nash equilibrium for such a game is called the {\em Multiclass Network Equilibrium Problem}.

\subsubsection*{Contribution}

Our results concern the case when the cost functions are affine and strictly increasing: for all $k\in K$ and $a\in A^k$, there exist $\alpha_a^k\in\mathbb{Q}_+\setminus\{0\}$ and $\beta_a^k\in\mathbb{Q}_+$ such that $c_a^k(x)=\alpha_a^kx+\beta_a^k$ for all $x\in\R_+$. 

First, we prove the existence of a polynomial algorithm solving the Multiclass Network Equilibrium Problem when the number of classes and the number of vertices are fixed. The core idea of the algorithm relies on properties of hyperplane arrangements. A corollary of this theorem is that the {\em parallel-link} case (graph with parallel arcs between two vertices) is polynomially solvable for a fixed number of classes. This special case, even with only two classes, does not seem to have been known before.

Second, we show that there exists a pivoting algorithm solving the problem. This algorithm, inspired by the classical Lemke algorithm solving linear complementarity problems, is reminiscent of the network simplex algorithm that solves the minimum cost flow problem, in the sense that we exploit the presence of a graph to build the pivoting algorithm. The experiments show its efficiency.
On our track, we extend slightly the notion of basis used in linear programming and linear complementarity programming to deal directly with unsigned variables (hence without replacing them by twice their number of signed variables). 

To our knowledge, these two algorithms are the first specially designed to solve this problem.

We emphasize that the exact complexity of the problem remains unknown. The fact that it can be modeled as a linear complementarity problem implies that it belongs to the so-called PPAD class. The PPAD class is a class of problems for which an object is sought, here an equilibrium, while being sure that the object exists by an {\em a priori} argument equivalent to the following one: in a graph without isolated vertices and whose vertices all have at most one predecessor and at most one successor, if there is a vertex with at most one neighbor, there is another such vertex. This class was defined by~\cite{P94} and contains complete problems. An example of a PPAD-complete problem is the problem of computing a mixed Nash equilibrium in a bimatrix game (\cite{CDT09}). We do not know whether the Multiclass Network Equilibrium Problem with affine costs is PPAD-complete or not.

\subsubsection*{Related works} The single-class case is polynomially solvable since as soon as the cost functions are nondecreasing, the problem turns out to be a convex optimization problem, see~\cite{Be56}. This case has already been mentioned at the beginning of the introduction.

We are not aware of any algorithm with a controlled complexity for solving the Multiclass Network Equilibrium Problem, even with affine cost functions. There are however some papers proposing practical approaches. In general, the proposed algorithm is a Gauss-Seidel type diagonalization method, which consists in sequentially fixing the flows for all classes but one and solving the resulting single-class problem by methods of convex programming, see \cite{F77,FS82,H88,MM88} for instance. For this method, a condition ensuring the convergence to an equilibrium is not always stated, and, when there is one, it requires that ``the interaction between the various users classes be relatively weak compared to the main effects (the latter translates a requirement that a complicated matrix norm be less than unity)''~\citep{MM88}. Such a condition does clearly not cover the case with affine cost functions. Another approach is proposed by \cite{MW04}.  For cost functions satisfying the ``nested monotonicity'' condition -- a notion developed by \cite{CC88} -- they design a descent method for which they are able to prove the convergence to a solution of the problem. However, we were not able to find any paper with an algorithm solving the problem when the costs are polynomial functions, or even affine functions.

We can also mention generalization of the Lemke algorithm -- see for instance \cite{AV11,AEP79,CF96,CPS92,E73,SPS12} --  but none of them is specially designed for solving our problem, nor exploits a graph structure of any kind.

\subsubsection*{Structure of the paper} In Section~\ref{sec:formulation}, we provide mathematical features of the Multiclass Network Equilibrium Problem. In particular, we show with an elementary proof how to write it as a linear complementarity problem. 

Section~\ref{sec:poly} is devoted to one of our main results, namely the existence of a polynomial algorithm when the number of vertices and the number of classes are fixed. This section is subdivided into fours subsections. The first subsection -- Section~\ref{subsec:main} -- states the result and gives a general description of the algorithm. We provide then a brief introduction to the concept of hyperplane arrangement, which is used in the proofs (Section~\ref{subsec:hyperplan}). The following two sections (Sections~\ref{subsec:determine} and~\ref{subsec:compute}) are devoted to the two parts of the proof.

Section~\ref{sec:lemke} is devoted to the network Lemke-like algorithm. The first subsection -- Section~\ref{subsec:opt} -- shows how to rewrite the linear complementary problem formulation as an optimization program. Section~\ref{subsec:tools} presents the notions underlying the algorithm. All these notions, like {\em basis}, {\em secondary ray}, {\em pivot}, and so on, are classical in the context of the Lemke algorithm. They require however to be redefined in order to be able to deal with the features of our optimization program. The algorithm is then described in Section~\ref{subsec:lemke}. Section~\ref{subsec:experiments} is devoted to the experiments and shows the efficiency of the proposed approach.

\begin{remark}
A preliminary version of Section~\ref{sec:lemke} has been presented at the conference WINE 2013 \citep{MP13Lemke}.
\end{remark}

\section{Mathematical properties of the equilibrium}\label{sec:formulation}

Let $\yy = (y_a)_{a\in A}$ be a flow. We define its {\em support} as the set of arcs with a positive flow: $$\supp(\yy) = \{ a \in A:\; y_a >0\}.$$ The cost of a minimum-cost $s^k$-$v$ path when the arc costs are given by the $c_a^k(\yy)$'s is denoted $\pi_v^k(\yy)$. 
We denote by $\pp^k(\yy)$ the vector $(\pi_v^k(\yy))_{v\in V^k}$. We define $\mincost^k(\yy)$ to be the set of arcs in $A^k$ that are on some minimum cost paths originating at $s^k$. Formally, we have $$\mincost^k(\yy)=\{a=(u,v)\in A^k:\;\pi_v^k(\yy)-\pi_u^k(\yy)=c_a^k(y_a)\}.$$ 

\begin{proposition}\label{prop:inclusions}
The multiflow $(\xx^k)_{k\in K}$ is an equilibrium multiflow if and only if $$\supp(\xx^k)\subseteq\mincost^k(\xx)\quad\mbox{for all $k\in K$}.$$
\end{proposition}

\begin{proof}
If some class $k$ players use an arc $a$ at equilibrium, it means that this arc is on a minimum-cost $s^k$-$t^k$ path. We have thus the inclusion $\supp(\xx^k)\subseteq\mincost^k(\xx)$.

Conversely, suppose that $\supp(\xx^k)\subseteq\mincost^k(\xx)$ for all $k$. Take any $s^k$-$t^k$ path $P$ chosen by a non-negligible amount of class $k$ players.  Each arc $a=(u,v)$ in this path is in the support of $\xx^k$, and thus is such that $\pi_v^k(\xx)-\pi_u^k(\xx)=c_a^k(x_a)$. The cost of the path is therefore $\pi_{t^k}^k(\xx)$, which implies that $P$ is a minimum-cost $s^k$-$t^k$ path. $(\xx^k)_{k\in K}$ is a multiflow equilibrium.
\end{proof}

With a similar proof, we can get an alternate formulation of the equilibrium. Consider the following system, where $\boldsymbol{b}=(b_v^k)$ is a given vector with $\sum_{v\in V^k}b_v^k=0$ for all $k$. 
\begin{equation}\label{pb:MNEP-gen}\tag{$MNEP_{gen}$}
\begin{array}{lr}
  \displaystyle{\sum_{a\in \delta^+(v)} x_a^k = \sum_{a\in \delta^-(v)} x_a^k + b_v^k} &  v\in V^k, k\in K \\ \\
c^k_{uv}(x_{uv}) + \pi^k_u - \pi^k_v - \mu^k_{uv} = 0 & (u,v)\in A^k, k\in K   \\ \\
  x_a^k \mu_a^k = 0 & a\in A^k, k\in K  \\ \\
  \pi_{s^k}^k=0 & k\in K\\ \\
  x_a^k \geq 0, \mu_a^k \geq 0, \pi_v^k \in \R &  v\in V^k, a \in A^k, k\in K .
  \end{array}
\end{equation}

Finding solutions for systems like~\eqref{pb:MNEP-gen} is a {\em complementarity program}, the word ``complementarity'' coming from the condition $x_a^k \mu_a^k = 0$ for all $(a,k)$ such that $a \in A^k$. 

\begin{proposition}\label{prop:equilibrium}
Suppose that $b_v^k=0$ for $v\in V^k\setminus\{s^k,t^k\}$, $b_{s^k}^k=\lambda(I^k)$, and $b_{t^k}^k=-\lambda(I^k)$ for all $k$. Then
$(\xx^k)_{k\in K}$ is an equilibrium multiflow if and only if there exist $\boldsymbol{\mu}^k\in\R_+^{A^k}$ and $\boldsymbol{\pi}^k\in\R^{V^k}$ for all $k$ such that $(\boldsymbol{x}^k,\boldsymbol{\mu}^k,\boldsymbol{\pi}^k)_{k\in K}$ is a solution of the complementarity program~\eqref{pb:MNEP-gen}.
\end{proposition}

\begin{proof}
Let $(\xx^k)_{k\in K}$ be a multiflow equilibrium. We define $\pi_v^k$ to be $\pi_v^k(\xx)$. Finally, $\mu^k_{uv}$ is defined to be $c^k_{uv}(x_{uv}) + \pi^k_u - \pi^k_v$ for all $k\in K$ and $(u,v)\in A^k$. This solution is a feasible solution of the program~\eqref{pb:MNEP-gen} (using Proposition~\ref{prop:inclusions} to get the complementary conditions).

Conversely, take a feasible solution of the program~\eqref{pb:MNEP-gen}. Let $P$ be any $s^k$-$t^k$ path. We have $\sum_{a\in P}c_a^k(x_a)=\pi_{t^k}^k+\sum_{a\in P}\mu_a^k$. Thus  $\sum_{a\in P}c_a^k(x_a)\geq\pi_{t^k}^k$, with equality when the path $P$ is in $\supp(\xx^k)$. Any $s^k$-$t^k$ path in $\supp(\xx^k)$ is thus a minimum-cost $s^k$-$t^k$ path. Hence $(\xx^k)_{k\in K}$ is an equilibrium multiflow.
\end{proof}

When the cost functions are affine $c_a^k(x)=\alpha_a^kx+\beta_a^k$, solving the Multiclass Network Equilibrium Problem amounts thus to solve the following linear complementarity problem

\begin{equation}\tag{$MNEP$}\label{pb:MNEP}
\begin{array}{lr}
\displaystyle{\sum_{a\in \delta^+(v)} x_a^k = \sum_{a\in \delta^-(v)} x_a^k + b_v^k} &  v\in V^k, k\in K\\ \\
\displaystyle{\alpha_{uv}^kx_{uv} + \pi^k_u - \pi^k_v - \mu^k_{uv} = - \beta_{uv}^k} &  (u,v) \in A^k, k\in K \\ \\
  x_a^k \mu_a^k = 0 &  a\in A^k, k\in K \\ \\
    \pi_{s^k}^k=0 & k\in K\\ \\
  x_a^k \geq 0, \mu_a^k \geq 0, \pi_v^k \in \R &  v\in V^k, a\in A^k, k\in K
  \end{array}
\end{equation}
with $b_v^k=0$ for $v\in V^k\setminus\{s^k,t^k\}$, and $b_{s^k}^k=\lambda(I^k)$ and $b_{t^k}^k=-\lambda(I^k)$ for all $k$.

\section{A polynomial algorithm}\label{sec:poly}

\subsection{The algorithm}\label{subsec:main}

We describe the algorithm solving the Multiclass Network Equilibrium Problem in polynomial time when the number of classes and the number of vertices are fixed. 

Let $\AA=\{(S^k)_{k\in K}:\;S^k\subseteq A^k\}$. The algorithm consists in two steps.
\begin{enumerate}
 \item It computes a set $\mathcal S  \subseteq \AA$ of polynomial size such that for any equilibrium multiflow $(\xx^k)_{k\in K}$, there is a $(S^k)_{k\in K} \in \mathcal S$ with $\supp(\xx^k) \subseteq S^k \subseteq \mincost^k(\xx)$ for all $k$. 
  \item It tests for every $(S^k)_{k\in K} \in \mathcal S$ whether there exists an equilibrium multiflow $(\xx^k)_{k\in K}$ with $\supp(\xx^k) \subseteq S^k \subseteq \mincost^k(\xx)$
  for all $k$, and compute it if it exists. 
\end{enumerate}
For fixed $|K|$ and $|V|$, each step can be done in polynomial time according respectively to Proposition~\ref{prop:poly_determine} and Proposition~\ref{prop:poly_compute}.

\begin{proposition} \label{prop:poly_determine} Assume $|K|$ and $|V|$ being fixed. We can determine in polynomial time a set $\mathcal S\subseteq \AA$ of polynomial size such that for any equilibrium multiflow $(\xx^k)_{k\in K}$, there is a $(S^k)_{k\in K} \in \mathcal S$ with $\supp(\xx^k) \subseteq S^k \subseteq \mincost^k(\xx)$ for all $k$.
\end{proposition}

Both the size of $\mathcal S$ and the time complexity to compute it are actually a $O((K^2|A|)^{K(|V|-1)})$.

In the next proposition, $|K|$ and $|V|$ are not required to be fixed. As we will see in the proof, it amounts to solve a system of linear equalities and inequalities, which is polynomially solvable thanks to the interior point method.

\begin{proposition} \label{prop:poly_compute}
Let $(S^k)_{k\in K}  \in \AA$. In polynomial time, we can \begin{itemize}
\item decide whether there exists an equilibrium multiflow $(\xx^k)_{k\in K}$ with $$\supp(\xx^k) \subseteq S^k \subseteq \mincost^k(\xx)$$ for all $k$, 
\item compute such a multiflow if it exists.
\end{itemize}
\end{proposition}

An equilibrium multiflow $(\xx^k)_{k\in K}$ is known to exist, see the Section ``Model'' of the Introduction. Thus, when the algorithm terminates, it has necessarily found an equilibrium.

To summarize, we have the following theorem. 

\begin{theorem}\label{thm:poly_main}
For a fixed number of classes and vertices, there exists an algorithm solving the Multiclass Network Equilibrium Problem with affine costs in polynomial time with respect to the number of arcs.
\end{theorem}

The complexity is $O\left((K^2|A|)^{K(|V|-1)}\right)$ times the complexity of solving a system of linear equalities and inequalities with $\sum_{k\in K}(|A^k|+|V^k|-1)$ variables.

\subsection{Preliminaries on hyperplane arrangements}\label{subsec:hyperplan}

A {\em hyperplane} $h$ in $\R^d$ is a $(d-1)$-dimensional subspace of $\R^d$. It partitions $\R^d$ into three regions: $h$ itself and the two open half-spaces having $h$ as boundary. We give an orientation for $h$ and note the two half-spaces $h^\oplus$ and $h^\ominus$, the former being on the positive side of $h$ and the latter one on the negative side. The closed half-spaces are denoted by $\overline{ h^\oplus} = h^\oplus \cup h$ and $\overline{ h^\ominus} = h^\ominus \cup h$. Given a finite set $H$ of hyperplanes, an {\em arrangement} is a partition of $\R^d$ into relatively open convex subsets, called {\em cells}. A $k$-cell is a cell of dimension $k$. A $0$-cell is called a point. The {\em hyperplane arrangement} $\mathcal A(H)$ associated to the set of hyperplanes $H$ is defined as follows. The $d$-cells are the connected components of $\R^d \setminus H$. For $0\leq k \leq d-1$, a $k$-{\em flat} is the intersection of exactly $d-k$ hyperplanes of $H$. Then, the $k$-cells of the arrangement are the connected components of $L \setminus \{ h \in H, L \nsubseteq h \}$ for every $k$-flat $L$.

Given an arrangement of $n$ hyperplanes, the number of $k$-cells is bounded by $$\sum_{i=0}^{k} \binom{d-i}{k-i} \binom{n}{d-i}.$$ The total number of cells is thus a $O(n^d)$. In a breakthrough paper, \cite{EdORSe86} proved that the set of cells (determined by the relative positions with respect to the hyperplans) can be computed in $O(n^d)$ as well, given the equations of the hyperplanes (and assuming that the coefficients involved in the equations are in $\mathbb Q$).

Further details on hyperplane arrangements can be found in \cite{Ed87} or \cite{Ma02} for example.

\subsection{Proof of Proposition~\ref{prop:poly_determine}}\label{subsec:determine}

For each class $k$, and each arc $a\in A^k$, we define the oriented half-space of $\prod_{j\in K}\R^{V^j\setminus\{s^j\}}$:
$$h_{a}^{k,\ominus} =  \left\{ \vec{\yy}=(y_v^k) \in\prod_{j\in K}\R^{V^j\setminus\{s^j\}}:\; y_v^{k}-y_u^k>  \beta_a^{k} \right\}.$$
For each class $k\neq k'$ and arc $a=(u,v)\in A^k\cap A^{k'}$, we define moreover the following oriented half-space, still of $\prod_{j\in K}\R^{V^j\setminus\{s^j\}}$:
$$h_a^{k,k',\ominus} = \left\{ \vec{\yy}=(y_v^j) \in\prod_{j\in K}\R^{V^j\setminus\{s^j\}}:\; \alpha_a^{k'} \left(y_v^{k}-y_u^k - \beta_a^{k}\right) > \alpha_a^k \left(y_v^{k'}-y_u^{k'} - \beta_a^{k'}\right)  \right\}.$$
We define the convex polyhedron $$P_a^k =\overline{h_{a}^{k,\ominus}}\cap \bigcap_{k'\neq k:\;A^k\cap A^{k'}\neq\emptyset} \overline{h_a^{k,k',\ominus}}.$$

The $P_a^k$'s have a useful property that links the cost at an equilibrium to the support. Let $\vec{\pp}(\xx)\in\prod_{k\in K}\R^{V^k\setminus\{s^k\}}$ be the vector $(\pp^k(\xx))_{k\in K}$.

\begin{lemma}\label{lem:1}
 Let $(\xx^k)_{k\in K}$ be an equilibrium multiflow. For any class $k$ and arc $a$,  if  $a \in \supp( \xx^k)$, then $\vec{\pp}(\xx)\in P_a^{k}$.
\end{lemma}

\begin{proof}
Let $a=(u,v) \in \supp(\xx^k)$. According to Proposition~\ref{prop:inclusions}, we have $$x_a = \frac{\pi_v^{k}(\xx)-\pi_u^{k}(\xx)-\beta_a^{k}}{\alpha_a^{k}}.$$ In particular, since $ x_a \geq 0$, we have $\pi_v^{k}(\xx)-\pi_u^{k}(\xx) \geq \beta_a^k$ and thus $ \vec{\pp}(\xx)\in \overline{h_{a}^{k,\ominus}} $. 

For any other class $k'$ such that $a\in A^{k'}$, we have 
$$\alpha_a^{k'} \left(\frac{\pi_v^{k}(\xx)-\pi_u^{k}(\xx)-\beta_a^k}{\alpha_a^k}\right) + \beta_a^{k'} \geq \pi_v^{k'}(\xx)-\pi_u^{k'}(\xx)$$ according to the definition of $\pp^{k'}(\xx)$. It implies that $\vec{\pp}(\xx) \in \overline{h_a^{k,k',\ominus}}.$ 

Therefore, $\vec{\pp}(\xx) \in P_a^k.$
\end{proof}

\bigskip

In order to prove the proposition, we consider the set of hyperplanes $$H = \left\{ h_a^{k,k'}:\; k\neq k' \in K, a\in A^k\cap A^{k'}\right\} \cup \left\{ h_{a}^{k}:\; k \in K, a\in  A^k\right\}.$$ We consider then the associated arrangement  $\mathcal{A}(H)$.

\begin{proof}[Proof of Proposition~\ref{prop:poly_determine}]
We start by building $\mathcal{A}(H)$. The number of cells and the time complexity to build them  are a $O\left( (K^2|A|)^{K(|V|-1)}\right)$ (see Section~\ref{subsec:hyperplan}).

Define the map $\varphi : \{\mbox{cells of }\mathcal{A}(H)\} \to \mathcal \AA$ in the following way: for every cell $P$ and class $k \in K$,
$$ \varphi(P)_k = \{ a \in A^k:\; P \cap P_a^k \neq \emptyset \}.$$ This map can easily be built in polynomial time. 

Let then $\mathcal S = \varphi(\{\mbox{cells of }\mathcal{A}(H)\})$. The size of $\mathcal S$ is at most $O\left( (K^2|A|)^{K(|V|-1)}\right)$.\\

It remains to show that for any equilibrium multiflow $(\xx^k)_{k\in K}$, there exists $(S^k)_{k\in K} \in \mathcal{S}$ such that $\supp(\xx^k) \subseteq S^k \subseteq \mincost^k(\xx)$ for all $k$.

Let $(\xx^k)_{k\in K}$ be an equilibrium multiflow. Since the cells of $\mathcal A (H)$ partition $\prod_{k\in K}\R^{V^k\setminus\{s^k\}}$, there is a cell $P_0$ such that $\vec{\pp}(\xx)\in P_0$. 
Let $k \in K$ and $a \in \supp( \xx^k)$. Lemma~\ref{lem:1} ensures that $\vec{\pp}(\xx)\in P_a^{k}$, and in particular that $P_0 \cap P_a^k \neq \emptyset$, i.e. $a \in \varphi(P_0)_k$. We have thus $\supp(\xx^k) \subseteq \varphi(P_0)_k$ for every $k\in K$. Defining $S^k = \varphi(P_0)_k$, we have $\supp(\xx^k)\subseteq S^k$ for all $k$, as required. \\

We prove now  that $S^k \subseteq \mincost^k(\xx)$ for all $k$. Consider a class $k$ and an arc $a\in S^k$. We have already proved that 
$\vec{\pp}(\xx)\in P_a^k$.

Suppose first that $x_a>0$. If $a \in \supp(\xx^k)$, Proposition~\ref{prop:inclusions} implies that $a \in \mincost^k(\xx)$. Otherwise, there is at least a class $k_0\neq k$ such that $a\in\supp(\xx^{k_0})$. Lemma \ref{lem:1} gives that $\vec{\pp}(\xx)\in P_a^{k_0}$. We have thus $\vec{\pp}(\xx)\in P_a^{k_0}\cap P_a^k$, which implies $\vec{\pp}(\xx)\in h_a^{k,k_0}$. This translates into
 $$   \alpha_a^k \left( \pi_v^{k_0}(\xx)-\pi_u^{k_0}(\xx) - \beta_a^{k_0}\right)  = \alpha_a^{k_0} (\pi_v^{k}(\xx)-\pi_u^{k}(\xx)-\beta_a^{k}),$$ i.e. 
$\alpha_a^kx_a+\beta_a^k=\pi_v^{k}(\xx)-\pi_u^{k}(\xx)$. Hence, $a \in \mincost^k(\xx)$.

Suppose then $x_a=0$. Since $\vec{\pp}(\xx)\in P_a^k$, we have in particular $\vec{\pp}(\xx)\in \overline{h_{a}^{k,\ominus}}$. It implies that $\pi_v^{k}(\xx)-\pi_u^{k}(\xx)\geq\beta_a^{k}$. The reverse inequality is a consequence of the definition of $\pp^{k}(\xx)$. We have thus $\pi_v^{k}(\xx)-\pi_u^{k}(\xx)=\beta_a^{k}$, which implies again $a \in \mincost^k(\xx)$.
\end{proof}

\subsection{Proof of Proposition~\ref{prop:poly_compute}}\label{subsec:compute}

\begin{proof}
There exists an equilibrium multiflow $(\xx^k)_{k\in K}$ with $\supp(\xx^k) \subseteq S^k \subseteq \mincost^k(\xx)$ for all $k$ if and only if there is a solution of the program~\eqref{pb:MNEP} with $\mu_a^k=0$ for all $k\in K$ and $a\in S^k$, and $x_a^k=0$ for all $k\in K$ and $a\notin S^k$. It gives rise to a system of linear equalities and inequalities with $\sum_{k\in K}(|A^k|+|V^k|-1)$ variables, which can be solved in polynomial time
by the interior point method (see \cite{Wr97} for example). 
\end{proof}

\begin{remark}
We can reduce the size of $\mathcal S$. We know without any computation that there are no solutions as soon as there is a class $k$ with $S^k = \emptyset$. It means that we can consider only the cells $P$ such that for every class $k$ there exists an arc $a$ with $P \cap P_a^k \neq \emptyset$. We can remove from $\mathcal A (H)$ the cells belonging to $$\bigcup_{k\in K} \bigcap_{a=(u,v) \in A^k} \left( \bigcup_{k'\neq k:\;A^k\cap A^{k'}=\emptyset}h_a^{k,k',\oplus} \cup h_{uv}^{k,\oplus}  \right).$$
However, this reduction is in general negligible with respect to the total size of $\mathcal S$.
\end{remark}

\section{A network Lemke-like algorithm}\label{sec:lemke}

\subsection{An optimization formulation}\label{subsec:opt}

Similarly as for the classical Lemke algorithm, we rewrite the problem as an optimization problem. It is the starting point of the algorithm. This problem is called the {\em Augmented Multiclass Network Equilibrium Problem}. Let $\ee=(e_a^k)$  be any vector defined for all $k\in K$ and $a\in A^k$. Consider the following optimization program.

\begin{equation}\tag{$AMNEP(\ee)$}\label{pb:AMNEP}
\begin{array}{rlr}\min & \omega \\
\mbox{s.t.} &  \displaystyle{\sum_{a\in \delta^+(v)} x_a^k = \sum_{a\in \delta^-(v)} x_a^k + b_v^k} &  k\in K,v \in V^k\\ \\
& \displaystyle{\alpha_{uv}^k\sum_{k'\in K}x^{k'}_{uv} + \pi^k_u - \pi^k_v - \mu^k_{uv} +e_{uv}^k\omega= - \beta_{uv}^k} & k\in K,(u,v) \in A^k  \\ \\
 & x_a^k \mu_a^k = 0 & k\in K, a\in A^k  \\ \\
  & \pi_{s^k}^k=0 & k\in K \\ \\
 & x_a^k \geq 0, \mu_a^k \geq 0, \omega\geq 0, \pi_v^k \in \R &  k\in K, a \in A^k, v\in V^k.
\end{array} \end{equation}

A key remark is 
\begin{quote}
{\em Solving \eqref{pb:MNEP} amounts to find an optimal solution of \eqref{pb:AMNEP}\\ with $\omega=0$.}
\end{quote}

Indeed, a solution with $\omega=0$ can easily be completed to provide a solution of \eqref{pb:MNEP}, and conversely, a solution of \eqref{pb:MNEP} provides a solution with $\omega=0$ of  \eqref{pb:AMNEP}. Some choices of $\ee$ allow to find easily feasible solutions to this program. In Section~\ref{subsec:tools}, $\ee$ will be chosen in such a way. 

We write the program~\eqref{pb:AMNEP} under the form
$$
\begin{array}{rl}\min & \omega \\
\mbox{s.t.} & 
 \overline{M}^{\ee}\left(\begin{array}{c}\boldsymbol{x} \\ \boldsymbol{\mu} \\ \omega\end{array}\right)+
\left(\begin{array}{c}\zero \\ M^T\end{array}\right)\boldsymbol{\pi}=\left(\begin{array}{c}\boldsymbol{b}\\-\boldsymbol{\beta}\end{array}\right) \\
& \boldsymbol{x}\cdot\boldsymbol{\mu}=0 \\
& \boldsymbol{x}\geq\zero,\,\boldsymbol{\mu}\geq\zero,\,\omega\geq 0,\,\boldsymbol{\pi}\in\prod_{k\in K}\R^{V^k\setminus\{s^k\}},
\end{array}
$$ where $\overline{M}^{\ee}$ and $C$ are defined as follows. (The matrix $\overline{M}^{\ee}$ is denoted with a superscript $\ee$ in order to emphasize its dependency on $\ee$).

We define $M=\diag((M^k)_{k\in K})$ where $M^k$ is the incidence matrix of the directed graph $(V^k,A^k)$ from which the $s^k$-row has been removed: 
$$M^k_{v,a}=\left\{ \begin{array}{ll} 1 & \mbox{ if $a\in\delta^+(v)$,} \\ -1 & \mbox{ if $a\in\delta^-(v)$,} \\ 0&\mbox{ otherwise}. \end{array}\right.$$

We also define $C^k=\diag((\alpha_a^k)_{a\in A^k})$ for $k\in K$, and then $C$ the real matrix $C=({\underbrace{(C^k,\cdots,C^k)}_{|K|\mbox{\tiny{ times}}}}{}_{k\in K})$. Then let $$\overline{M}^{\ee}= \left(\begin{array}{ccc} M & \zero & \zero \\ C & -I & \ee\end{array}\right).$$ 

For $k\in K$, the matrix $M^k$ has $|V^k|-1$ rows and $|A^k|$ columns, while $C^k$ is a square matrix with $|A^k|$ rows and columns. Then the whole matrix $\overline{M}^{\ee}$ has $\sum_{k\in K}(|A^k|+|V^k|-1)$ rows and $2\left(\sum_{k\in K}|A^k|\right)+1$ columns.

\subsection{Bases, pivots, and rays}\label{subsec:tools}

\subsubsection{Bases}\label{subsec:bases} We define $\X$ and $\M$ to be two disjoint copies of $\{(a,k):\,k\in K,\,a\in A^k\}$. We denote by $\phi^x(a,k)$ (resp. $\phi^{\mu}(a,k)$) the element of $\X$ (resp. $\M$) corresponding to $(a,k)$. The set $\X$ models the set of all possible indices for the `$x$' variables and $\M$ the set of all possible indices for the `$\mu$' variables for the program~\eqref{pb:AMNEP}. We consider moreover a dummy element $o$ as the index for the `$\omega$' variable.

We define a {\em basis} for the program~\eqref{pb:AMNEP} to be a subset $B$ of the set $\X\cup\M\cup\{o\}$ such that the square matrix of size $\sum_{k\in K}\left(|A^k|+|V^k|-1\right)$ defined by
$$\left(\begin{array}{c|c}\overline{M}^{\ee}_B & \begin{array}{c} \zero \\ M^T \end{array}\end{array}\right)$$ is nonsingular. Note that this definition is not standard. In general, a basis is defined in this way but without the submatrix $\left(\begin{array}{c} \zero \\ M^T \end{array}\right)$ corresponding to the `$\pi$' columns. We use this definition in order to be able to deal directly with the unsigned variables `$\pi$'. We will see that this approach is natural (and could be used for linear programming as well). However, we are not aware of a previous use of such an approach.

As a consequence of this definition, since $M^T$ has $\sum_{k\in K}(|V^k|-1)$ columns, a basis is always of cardinality $\sum_{k\in K}|A^k|$.

\begin{remark}\label{rem:size}
 In particular, since the matrix is nonsingular and since $M^T$ has $\sum_{k\in K}|A^k|$ rows, the first $\sum_{k\in K}(|V^k|-1)$ rows of $\overline{M}^{\ee}_B$ have each a nonzero entry. This property is used below, especially in the proof of Lemma~\ref{lem:nosecondary}.
\end{remark}

The following additional notation is useful: given a subset $Z\subseteq\X\cup\M\cup\{o\}$, we denote by $Z^x$ the set $\left(\phi^{x}\right)^{-1}(Z\cap\X)$ and by $Z^{\mu}$ the set $\left(\phi^{\mu}\right)^{-1}(Z\cap\M)$.
In other words, $(a,k)$ is in $Z^x$ if and only if $\phi^x(a,k)$ is in $Z$, and similarly for $Z^{\mu}$.

\subsubsection{Basic solutions and non-degeneracy}\label{subsec:basic}
Let $B$ be a basis. If it contains $o$, the unique solution 
$(\bar{\boldsymbol{x}},\bar{\boldsymbol{\mu}},\bar{\omega},\bar{\boldsymbol{\pi}})$
of 
\begin{equation}\label{eq:basicsol}
\left\{\begin{array}{l}
\left(\begin{array}{c|c}\overline{M}^{\ee}_B & \begin{array}{c} \zero \\ M^T \end{array}\end{array}\right)\left(\begin{array}{c}\boldsymbol{x}_{B^x} \\ \boldsymbol{\mu}_{B^{\mu}} \\ \omega \\ \boldsymbol{\pi}\end{array}\right)=\left(\begin{array}{c}\boldsymbol{b} \\ -\boldsymbol{\beta} \end{array}\right) \\
x_a^k = 0\quad\mbox{ for all $(a,k)\notin B^x$} \\
\mu_a^k = 0\quad\mbox{ for all $(a,k)\notin B^{\mu}$}.
\end{array}\right.
\end{equation} is called the {\em basic solution} associated to $B$. If $B$ does not contain $o$, we define similarly its associated {\em basic solution}. It is the unique solution $(\bar{\boldsymbol{x}},\bar{\boldsymbol{\mu}},\bar{\omega},\bar{\boldsymbol{\pi}})$ of
\begin{equation}\label{eq:basicsol_opt}
\left\{\begin{array}{l}
\left(\begin{array}{c|c}\overline{M}^{\ee}_B & \begin{array}{c} \zero \\ M^T \end{array}\end{array}\right)\left(\begin{array}{c}\boldsymbol{x}_{B^x} \\ \boldsymbol{\mu}_{B^{\mu}} \\ \boldsymbol{\pi}\end{array}\right)=\left(\begin{array}{c}\boldsymbol{b} \\ -\boldsymbol{\beta} \end{array}\right) \\
x_a^k = 0\quad\mbox{ for all $(a,k)\notin B^x$}\\
\mu_a^k = 0\quad\mbox{ for all $(a,k)\notin B^{\mu}$} \\
\omega = 0.
\end{array}\right.
\end{equation} 
A basis is said to be {\em feasible} if the associated basic solution is such that $\bar{\boldsymbol{x}},\bar{\boldsymbol{\mu}},\bar{\omega}\geq 0$.  \\

The program~\eqref{pb:AMNEP} is said to {\em satisfy the non-degeneracy assumption} if, for any feasible basis $B$, the associated basic solution $(\bar{\boldsymbol{x}},\bar{\boldsymbol{\mu}},\bar{\omega},\bar{\boldsymbol{\pi}})$ is such that 
$$\left((a,k)\in B^x\Rightarrow\bar{x}_a^k>0\right)\mbox{ and }\left((a,k)\in B^{\mu}\Rightarrow\bar{\mu}_a^k>0\right).$$ Note that if we had defined the vector $\boldsymbol{b}$ to be $0$ on all vertices $v\notin\{s^k,t^k\}$ -- as it is required by the original formulation of the Multiclass Network Equilibrium Problem -- the program would not in general satisfy the non-degeneracy assumption. Our network Lemke-like algorithm actually solves the program~\eqref{pb:AMNEP} under the non-degeneracy assumption, but, as it will be explained in Section~\ref{subsec:lemke}, it can be used to solve the degenerate case as well -- and thus the original formulation when the costs are affine -- via a perturbation argument.

An example of a basis for which the assumption fails to be satisfied is the basis $B^{ini}$ defined in Section~\ref{subsec:init}.  Remark~\ref{rem:degeneracy} in that section details the example.

\subsubsection{Pivots and polytope}
The following lemmas are key results that eventually lead to the Lemke-like algorithm. They are classical for the usual definition of bases. Since we have extended the definition, we have to prove that they still hold.

\begin{lemma}\label{lem:pivotout}
Let $B$ be a feasible basis for the program~\eqref{pb:AMNEP} and assume non-degeneracy. Let $i$ be an index in $\X\cup\M\cup\{o\}\setminus B$. Then there is at most one feasible basis $B'\neq B$ in the set $B\cup\{i\}$.
\end{lemma}
\begin{proof}
Let $(\bar{\boldsymbol{x}},\bar{\boldsymbol{\mu}},\bar{\omega},\bar{\boldsymbol{\pi}})$ be the basic solution associated to $B$ and let $Y=B\cup\{i\}$.
The set of solutions
 $$\left\{\begin{array}{l}
\left(\begin{array}{c|c}\overline{M}^{\ee}_Y & \begin{array}{c} \zero \\ M^T \end{array}\end{array}\right)\left(\begin{array}{c}\boldsymbol{x}_{Y^x} \\ \boldsymbol{\mu}_{Y^{\mu}} \\ \omega \\ \boldsymbol{\pi}\end{array}\right)=\left(\begin{array}{c}\boldsymbol{b} \\ -\boldsymbol{\beta} \end{array}\right) \\
x_a^k = 0\quad\mbox{ for all $(a,k)\notin Y^x$} \\
\mu_a^k = 0\quad\mbox{ for all $(a,k)\notin Y^{\mu}$}\end{array}\right.$$ is a one-dimensional line in $\R\times\prod_{k\in K}\left((\R^2)^{A^k}\times\R^{V^k\setminus\{s^k\}}\right)$ (the space of all variables) and passing through $(\bar{\boldsymbol{x}},\bar{\boldsymbol{\mu}},\bar{\omega},\bar{\boldsymbol{\pi}})$. The bases in $Y$ correspond to intersections of this line with the boundary of $$Q=\{(\boldsymbol{x},\boldsymbol{\mu},\omega,\boldsymbol{\pi}):\,x_a^k\geq 0, \mu_a^k\geq 0, \omega\geq 0,\mbox{ for all $k\in K$ and $a\in A^k$}\}.$$ This latter set being convex (it is a polyhedron), the line intersects at most twice its boundary under the non-degeneracy assumption.
\end{proof}

The operation consisting in computing $B'$ given $B$ and the {\em entering index} $i$ is called the {\em pivot operation}. If we are able to determine an index in $\X\cup\M\cup\{o\}\setminus B$ for any basis $B$, Lemma~\ref{lem:pivotout} leads to a ``pivoting'' algorithm. At each step, we have a current basis $B^{curr}$, we determine the entering index $i$, and we compute the new basis in $B^{curr}\cup\{i\}$, if it exists, which becomes the new current basis $B^{curr}$; and so on. The next lemma allows us to characterize situations where there is no new basis, i.e. situations for which the algorithm gets stuck.

The feasible solutions of \eqref{pb:AMNEP} belong to the polytope 
\begin{align*}
 \mathcal{P}(\ee)=\left\{(\boldsymbol{x},\boldsymbol{\mu},\omega,\boldsymbol{\pi}):\,
 \overline{M}^{\ee}\left(\begin{array}{c}\boldsymbol{x} \\ \boldsymbol{\mu} \\ \omega\end{array}\right)+\left(\begin{array}{c}\zero \\ M^T\end{array}\right)\boldsymbol{\pi}=\left(\begin{array}{c}\boldsymbol{b}\\-\boldsymbol{\beta}\end{array}\right),\, \right. \\
  \boldsymbol{x}\geq\zero,\,\boldsymbol{\mu}\geq\zero,\,\boldsymbol{\pi}\geq\zero,\,\omega\in\R_+ \Bigg\}. 
\end{align*}

\begin{lemma}\label{lem:infiniteray}
Let $B$ be a feasible basis for the program~\eqref{pb:AMNEP} and assume non-degeneracy. Let $i$ be an index in $\X\cup\M\cup\{o\}\setminus B$. If there is no feasible basis $B'\neq B$ in the set $B\cup\{i\}$, then the polytope $\mathcal{P}(\ee)$ contains an infinite ray originating at the basic solution associated to $B$.
\end{lemma}
\begin{proof}
The proof is similar as the one of Lemma~\ref{lem:pivotout}, of which we take the same notions and notations. If $B$ is the only feasible basis, then the line intersects the boundary of $Q$ exactly once. Because of the non-degeneracy assumption, it implies that there is an infinite ray originating at $(\bar{\boldsymbol{x}},\bar{\boldsymbol{\mu}},\bar{\omega},\bar{\boldsymbol{\pi}})$ and whose points are all feasible.
\end{proof}

\subsubsection{Complementarity and twin indices}\label{subsec:comp}
A basis $B$ is said to be {\em complementary} if for every $(a,k)$ with $a\in A^k$, we have $(a,k)\notin B^x$ or $(a,k)\notin B^{\mu}$: for each $(a,k)$, one of the components $x_a^k$ or  $\mu_a^k$ is not activated in the basic solution. In case of non-degeneracy, it coincides with the condition $\boldsymbol{x}\cdot\boldsymbol{\mu}=0$. An important point to be noted for a complementary basis $B$ is that if $o\in B$, then there is $(a_0,k_0)$ with $a_0\in A^{k_0}$ such that 
\begin{itemize}
\item $(a_0,k_0)\notin B^x$ and $(a_0,k_0)\notin B^{\mu}$, and
\item for all $(a,k)\neq(a_0,k_0)$ with $a\in A^k$, exactly one of the relations $(a,k)\in B^x$ and $(a,k)\in B^{\mu}$ is satisfied.
\end{itemize} 
This is a direct consequence of the fact that there are exactly $\sum_{k\in K}|A^k|$ elements in a basis and that each $(a,k)$ is not present in at least one of $B^x$ and $B^{\mu}$. In case of non-degeneracy, this point amounts to say that $x_a^k=0$ or $\mu_a^k=0$ for all $(a,k)$ with $a\in A^k$ and that there is exactly one such pair, denoted $(a_0,k_0)$, such that both are equal to $0$. 

We say that $\phi^x(a_0,k_0)$ and $\phi^{\mu}(a_0,k_0)$ for such $(a_0,k_0)$ are the {\em twin indices}.

\subsubsection{Initial feasible basis}\label{subsec:init}

A good choice of $\ee$ gives an easily computable initial feasible complementary basis to the program~\eqref{pb:AMNEP}. 

An {\em $s$-arborescence} in a directed graph is a spanning tree rooted at $s$ that has a directed path from $s$ to any vertex of the graph.
We arbitrarily define a collection  $\mathcal{T}=(T^k)_{k\in K}$ where $T^k\subseteq A^k$ is an $s^k$-arborescence of $(V^k,A^k)$. Then the vector $\ee=(e_a^k)_{k\in K, a\in A^k}$ is chosen with the help of $\mathcal T$ by 
\begin{equation}\label{eq:defe}e_a^k=\left\{\begin{array}{ll} 1 & \mbox{if $a\notin T^k$} \\ 0 & \mbox{otherwise}.\end{array}\right.\end{equation} 

\begin{lemma}\label{lem:initbasis}
 Let the set of indices $Y\subseteq\X\cup\M\cup\{o\}$ be defined by 
$$Y=\{\phi^x(a,k):\,a\in T^k, k\in K\}\cup\{\phi^{\mu}(a,k):\,a\in A^k\setminus T^k, k\in K\}\cup\{o\}.$$
Then, one of the following situations occurs:
\begin{itemize}
 \item[$\bullet$] $Y\setminus \{o\}$ is a complementary feasible basis providing an optimal solution of the program \eqref{pb:AMNEP} with $\omega=0$.
 \item[$\bullet$] There exists $(a_0,k_0)$ such that $B^{ini}=Y\setminus\{\phi^{\mu}(a_0,k_0)\}$ is a feasible complementary basis for the program~\eqref{pb:AMNEP}.
\end{itemize}
\end{lemma}

\begin{proof}
 The subset $Y$ has cardinality $\sum_{k\in K}|A^k|+1$. To show that $Y$ contains a feasible complementary basis, we proceed by studying the solutions of the system
 
\begin{equation}\tag{$S^{\ee}$}\label{eq:Y}
\left\{\begin{array}{l}
\left(\begin{array}{c|c}\overline{M}^{\ee}_Y & \begin{array}{c} \zero \\ M^T \end{array}\end{array}\right)\left(\begin{array}{c}\boldsymbol{x}_{Y^x} \\ \boldsymbol{\mu}_{Y^{\mu}} \\ \omega \\ \boldsymbol{\pi}\end{array}\right)=\left(\begin{array}{c}\boldsymbol{b} \\ -\boldsymbol{\beta} \end{array}\right) \\
x_a^k = 0\quad\mbox{ for all $(a,k)\notin Y^x$} \\
\mu_a^k = 0\quad\mbox{ for all $(a,k)\notin Y^{\mu}$}.
\end{array}\right.
\end{equation} 
It is convenient to rewrite the system~\eqref{eq:Y} in the following form.
\begin{align}\mbox{For}&\mbox{ all } k\in K, \nonumber\\ 
&\left\{\begin{array}{ll}
M_{T^k}^k x_{T^k}^k = b^k  &    \\ 
\alpha_{uv}^k\displaystyle{\sum_{k'\in K}x^{k'}_{uv} + \pi^k_u - \pi^k_v - \mu^k_{uv} +e_{uv}^k\omega= - \beta_{uv}^k} & \mbox{ for all } (u,v) \in A^k  \\ 
 x_a^k = 0 & \mbox{ for all }  a\notin T^k   \\
 \mu_a^k = 0 & \mbox{ for all }   a\in T^k. 
\end{array}\right.\label{eq:S}\end{align}

The matrix $M_{T^k}^k$ is nonsingular (see the book by \cite{AMO93}). It gives a unique solution $x_{T^k}^k$ of the first equation of \eqref{eq:S}, and since $x_a^k = 0$ for $a\notin T^k$, we get a unique solution  $\boldsymbol{x}$ to system~\eqref{eq:Y}.

We look now at the second equation of \eqref{eq:S} for $k$ and $(u,v)$ such that $(u,v)\in T^k$. We get that any solution of system~\eqref{eq:Y} satisfies the equalities
$$\alpha_{uv}^k\sum_{k'\in K}x_{uv}^{k'} + \pi^k_u - \pi^k_v = - \beta_{uv}^k, \quad \mbox{ for all $k\in K$ and $(u,v)\in T^k$}.$$
Indeed, if $(u,v)\in T^k$, we have $e_{uv}^k=0$ and $\mu_{uv}^k=0$. Recall that we defined $\pi_{s^k}^k=0$. Since $T^k$ is a spanning tree of $(V^k,A^k)$ for all $k$, these equations completely determine $\boldsymbol{\pi}$.

We look then at the second equation of \eqref{eq:S}, this time for $k$ and $(u,v)$ such that $(u,v)\notin T^k$. We get that any solution of system~\eqref{eq:Y} satisfies the equalities
\begin{equation}\label{eq:mu}
\alpha_{uv}^k\sum_{k'\neq k}x_{uv}^{k'} -\mu_{uv}^k+\omega+\pi^k_u-\pi^k_v=-\beta_{uv}^k, \quad \mbox{ for all $k\in K$ and $(u,v)\notin T^k$}. 
\end{equation} Indeed, if $(u,v)\notin T^k$, we have $e_{uv}^k=1$ and $x_{uv}^k=0$.

If $\alpha_{uv}^k x_{uv}+\beta_{uv}^k+\pi^k_u-\pi^k_v\geq 0$ for all $k\in K$ and $(u,v)\notin T^k$, then we have an optimal solution of the program~\eqref{pb:AMNEP} with $\omega=0$, and we get the first point of Lemma~\ref{lem:initbasis}. We can thus assume that $\alpha_{uv}^k x_{uv}+\beta_{uv}^k+\pi^k_u-\pi^k_v<0$ for at least one triple $u,v,k$. Let $u_0,v_0,k_0$ be such a triple minimizing $\alpha_{uv}^k x_{uv}+\beta_{uv}^k+\pi^k_u-\pi^k_v$ and let $a_0=(u_0,v_0)$. Note that Equation~\eqref{eq:mu} implies that 
\begin{equation}\label{eq:mupositive}
 \mu_{uv}^k \geq \mu_{u_0v_0}^{k_0}, \quad \mbox{ for all $k\in K$ and $(u,v)\notin T^k$}.
\end{equation}

We finish the proof by showing that $B^{ini}$, defined as $Y\setminus\{\phi^{\mu}(a_0,k_0)\}$, is a feasible complementary basis for the program~\eqref{pb:AMNEP}. For $B^{ini}$, system~\eqref{eq:basicsol} has a unique solution. Indeed, the first part of the proof devoted to the solving of~\eqref{eq:Y} has shown that $\boldsymbol{x}$ and $\boldsymbol{\pi}$ are uniquely determined, without having to compute the values of the $\mu_a^k$'s. By definition of $(a_0,k_0)$, since $\phi^{\mu}(a_0,k_0)$ is not in $B^{ini}$, we have $$\mu_{u_0v_0}^{k_0}=0\quad\mbox{and}\quad\omega = -\alpha_{u_0v_0}^{k_0} x_{u_0v_0}- \beta_{u_0v_0}^{k_0}-\pi^{k_0}_{u_0}+\pi^{k_0}_{v_0}.$$ Finally, Equation~(\ref{eq:mu}) determines the values of the $\mu_{uv}^k$ for $k\in K$ and $(u,v) \notin T^k$, and Equation~(\ref{eq:mupositive}) ensures that these values are nonnegative. Therefore, $B^{ini}$ is a basis, and it is feasible because all $x_a^k$ and $\mu_a^k$ in the solution are nonnegative. Furthermore, for each $(a,k)$ with $a\in A^k$, at least one of $\phi^x(a,k)$ and $\phi^{\mu}(a,k)$ is not in $B^{ini}$. Hence, the subset $B^{ini}$ is a feasible complementary basis.
\end{proof}
 We emphasize that $B^{ini}$ depends on the chosen collection $\mathcal{T}$ of arborescences. 
Note that the basis $B^{ini}$ is polynomially computable.

\begin{remark}\label{rem:base}
A short examination of the proof makes clear that the following claim is true: {\em Assuming non-degeneracy, if $B$ is a feasible basis such that $B^x=\{(a,k):\,a\in T^k,\,k\in K\}$, then $B=B^{ini}$.}
The fact that the $T^k$ are arborescences fixes completely $\boldsymbol{x}$, and then $\boldsymbol{\pi}$. The fact that $B$ is a feasible basis forces $\omega$ to be equal to the maximal value of $-\alpha_{uv}^k x_{uv}- \beta_{uv}^{k}-\pi^{k}_{u}+\pi^{k}_{v}$ (except of course if this value is nonpositive, in which case we have already solved our problem), which in turn fixes the values of the $\mu_{uv}^k$.
\end{remark}

\begin{remark}\label{rem:degeneracy}
As already announced in Section~\ref{subsec:basic}, if we had defined the vector $\boldsymbol{b}$ to be $0$ on all vertices $v\notin\{s^k,t^k\}$, the problem would not satisfy the non-degeneracy assumption as soon as there is $k\in K$ such that $T^k$ has a vertex of degree $3$ (which happens when $(V^k,A^k)$ has no Hamiltonian path). In this case, the basis $B^{ini}$ shows that the problem is degenerate. Since the unique solution $\boldsymbol{x}^k_{T^k}$ of $M_{T^k}^k \boldsymbol{x}^k_{T^k} = \boldsymbol{b}^k$ consists in sending the whole demand on the unique path in $T^k$ from $s^k$ to $t^k$, we have for all arcs $a \in T^k$ not belonging to this path $x_a^k =0$ while $(a,k)\in B^{ini,x}$.
\end{remark}

\subsubsection{No secondary ray} \label{subsec:ray}

Let $(\bar{\boldsymbol{x}}^{ini},\bar{\boldsymbol{\mu}}^{ini},\bar{\omega}^{ini},\bar{\boldsymbol{\pi}}^{ini})$ be the feasible basic solution associated to the initial basis $B^{ini}$, computed according to Lemma~\ref{lem:initbasis} and with $\ee$ given by Equation~\eqref{eq:defe}. 
The following inifinite ray
$$\rho^{ini}=\left\{(\bar{\boldsymbol{x}}^{ini},\bar{\boldsymbol{\mu}}^{ini},\bar{\omega}^{ini},\bar{\boldsymbol{\pi}}^{ini})+t(\zero,\ee,1,\zero):\, t\geq0\right\},$$ has all its points in $\mathcal{P}(\ee)$. This ray with direction $(\zero,\ee,1,\zero)$ is called the {\em primary ray}. In the terminology of the Lemke algorithm, another infinite ray originating at a solution associated to a feasible complementary basis is called a {\em secondary ray}. Recall that we defined $\pi_{s^k}^k=0$ for all $k \in K$ in Section~\ref{sec:formulation} (otherwise we would have a trivial secondary ray). System~\eqref{pb:AMNEP} has no secondary ray for the chosen $\ee$.

\begin{lemma}\label{lem:nosecondary}
Let $\ee$ be defined by Equation~\eqref{eq:defe}. Under the non-degeneracy assumption, there is no secondary ray in $\mathcal{P}(\ee)$.
\end{lemma}

\begin{proof}
Suppose that $\mathcal{P}(\ee)$ contains an infinite ray $$\rho=\left\{(\bar{\boldsymbol{x}},\bar{\boldsymbol{\mu}},\bar{\omega},\bar{\boldsymbol{\pi}})+t(\boldsymbol{x}^{dir},\boldsymbol{\mu}^{dir},\omega^{dir},\boldsymbol{\pi}^{dir}):\, t\geq0\right\},$$ where $(\bar{\boldsymbol{x}},\bar{\boldsymbol{\mu}},\bar{\omega},\bar{\boldsymbol{\pi}})$ is a feasible complementary basic solution associated to a basis $B$. \\

We first show that $\boldsymbol{x}^{dir}=0$. For a contradiction, suppose that it is not the case and let $k$ be such that $\boldsymbol{x}^{dir,k}$ is not zero. Since the points of $\rho$ must satisfy the system~\eqref{pb:AMNEP} for all $t\geq 0$, we have that $(\boldsymbol{x}^{dir},\boldsymbol{\mu}^{dir},\omega^{dir},\boldsymbol{\pi}^{dir})$ must satisfy for all $v\in V^k$
$$\sum_{a\in \delta^+(v)} x_a^{dir,k} = \sum_{a\in \delta^-(v)} x_a^{dir,k},$$ which shows that $\boldsymbol{x}^k$ is a circulation in the directed graph $(V^k,A^k)$. Moreover, we must have for all $(u,v)\in A^k$
\begin{equation}\label{eq:cost}
\begin{array}{c}
\displaystyle{\alpha_{uv}^k\sum_{k'\in K}x^{dir,k'}_{uv} + \pi^{dir,k}_u - \pi^{dir,k}_v - \mu^{dir,k}_{uv} +e_{uv}^k\omega^{dir}=0}.
\end{array}
\end{equation} where we have $\pi_{s^k}^{dir,k}=0$ since $\pi_{s^ k}^k=0$ for any feasible solution of~\eqref{pb:AMNEP}, see Section~\ref{sec:formulation}. The following relations must also be satisfied:
\begin{equation}\label{eq:xmu}
\boldsymbol{x}^{dir}\cdot\boldsymbol{\mu}^{dir}=0,
\end{equation} and
\begin{equation}\label{eq:omega}
\boldsymbol{x}^{dir}\geq\boldsymbol{0},\boldsymbol{\mu}^{dir}\geq\boldsymbol{0},\omega^{dir}\geq 0.
\end{equation} 
Take now any circuit $C$ in $D=(V,A)$ in the support of $\boldsymbol{x}^{dir,k}$. Since we have supposed that $\boldsymbol{x}^{dir,k}$ is not zero and since it is a circulation, such a circuit necessarily exists. According to Equations~\eqref{eq:xmu} and~\eqref{eq:omega}, we have $\mu_a^{dir,k}=0$ for each $a\in C$.
The sum $\sum_{a\in C}e_a^k$ is nonzero since no tree $T^k$ can contain all arcs in $C$. Summing Equation~\eqref{eq:cost} for all arcs in $C$, we get $$\omega^{dir}=-\frac{\sum_{a\in C}\alpha_a^k\sum_{k'\in K}x_a^{dir,k'}}{\sum_{a\in C}e_a^k}<0.$$ It is in contradiction with Equation~\eqref{eq:omega}. It implies that $x_a^{dir,k}=0$ for all $k\in K$ and $a\in A^k$. \\

We show now that $\boldsymbol{\pi}^{dir}=0$. We start by noting that Equation~\eqref{eq:cost} becomes $$\pi_u^{dir,k}-\pi_v^{dir,k}-\mu_{uv}^{dir,k}=0, \quad \mbox{ for all $k\in K$ and $(u,v)\in T^k$}.$$ Since $T^k$ is an $s^k$-arborescence, we have $0=\pi_{s^k}^{dir,k}\geq\pi_v^{dir,k}$ for all $v\in V^k$, according to Equation~\eqref{eq:omega}. 

Define now $F^k$ to be the set of arcs $a\in A^k$ such that $(a,k)\in B^x$. Using Remark~\ref{rem:size} of Section~\ref{subsec:bases}, $\overline{M}^{\ee}_B$ has a nonzero entry on each of its first $\sum_{k\in K}(|V^k|-1)$ rows, which implies that the set $F^k$ spans all vertices in $V^k\setminus\{s^k\}$.%

According to the non-degeneracy assumption, $\bar{x}_a^k$ is nonzero on all arcs of $F^k$. The complementarity condition for all points of the ray give that $\bar{\boldsymbol{x}}\cdot\boldsymbol{\mu}^{dir}+\boldsymbol{x}^{dir}\cdot\bar{\boldsymbol{\mu}}=0$, and since $\boldsymbol{x}^{dir} = \zero$, we have 
$\bar{\boldsymbol{x}}\cdot\boldsymbol{\mu}^{dir}=0$. Hence $\mu_{uv}^{dir,k}=0$ for all $(u,v)\in F^k$, and Equation~\eqref{eq:cost} becomes 
\begin{equation}\label{eq:cost_bis}\pi_u^{dir,k}-\pi_v^{dir,k}+e_{uv}^k\omega^{dir}=0 \quad \mbox{ for all $k\in K$ and $(u,v)\in F^k$}.\end{equation}  Thus, according to Equation~\eqref{eq:omega}, we have $0=\pi_{s^k}^{dir,k}\leq\pi_v^{dir,k}$ for all $v\in V^k$. Since we have already shown the reverse inequality, we have $\pi_v^{dir,k}=0$ for all $v\in V^k$. \\

Now, if $T^k\neq F^k$ for at least one $k$, we get the existence of an arc $(u,v)\in F^k$ for which $e_{uv}^k=1$, while $\pi_u^{dir,k}=\pi_v^{dir,k}=0$. Equation~\eqref{eq:cost_bis} implies then that $\omega^{dir}=0$. Still using $\boldsymbol{x}^{dir}=\zero$, we get then, again with the help of Equation~\eqref{eq:cost}, that $\boldsymbol{\mu}^{dir}=\zero$, which contradicts the fact that $\rho$ is an infinite ray.

Therefore, we have $T^k=F^k$ for all $k$. Using Remark~\ref{rem:base} of Section~\ref{subsec:init}, we are at the initial basic solution: $B=B^{ini}$. According to Equation~\eqref{eq:cost}, and since $\boldsymbol{x}^{dir}=\zero$ and $\boldsymbol{\pi}^{dir}=\zero$, we have $\mu_{uv}^{dir,k}=e_{uv}^k\omega^{dir}$ for all $k\in K$ and $(u,v)\in A^k$. Thus $(\boldsymbol{x}^{dir},\boldsymbol{\mu}^{dir},\omega^{dir},\boldsymbol{\pi}^{dir})=\omega^{dir}(\zero,\ee,1,\zero)$ for $\omega^{dir}\geq 0$,  and $\rho$ is necessarily the primary ray $\rho^{ini}$. 

Then there is no secondary ray, as required.\end{proof}

\subsubsection{A Lemke-like algorithm}\label{sec:algo}
Assuming non-degeneracy, the combination of Lemma~\ref{lem:pivotout} and the point explicited in Section~\ref{subsec:comp} gives rise to a Lemke-like algorithm. Two feasible complementary bases $B$ and $B'$ are said to be {\em neighbors} if $B'$ can be obtained from $B$ by a pivot operation using one of the twin indices as an entering index, see Section~\ref{subsec:comp}. Note that is is a symmetrical notion: $B$ can then also be obtained from $B'$ by a similar pivot operation. The abstract graph whose vertices are the feasible complementary bases and whose edges connect neighbor bases is thus a collection of paths and cycles. According to Lemma~\ref{lem:initbasis}, we can find in polynomial time an initial feasible complementary basis for~\eqref{pb:AMNEP} with the chosen vector $\ee$. This initial basis has exactly one neighbor according to Lemma~\ref{lem:infiniteray} since there is a primary ray and no secondary ray (Lemma~\ref{lem:nosecondary}).

Algorithm~\ref{algo:pivot_compl} explains how to follow the path starting at this initial feasible complementary basis. Function \texttt{EnteringIndex}$(B,i')$ is defined for a feasible complementary basis $B$ and an index $i'\notin B$ being a twin index of $B$ and computes the other twin index $i\neq i'$. Function \texttt{LeavingIndex}$(B,i)$ is defined for a feasible complementary basis $B$ and an index $i\notin B$ and computes the unique index $j\neq i$ such that $B\cup\{i\}\setminus\{j\}$ is a feasible complementary basis (see Lemma~\ref{lem:pivotout}).

Since there is no secondary ray (Lemma~\ref{lem:nosecondary}), a pivot operation is possible because of Lemma~\ref{lem:infiniteray} as long as there are twin indices. By finiteness, a component in the abstract graph having an endpoint necessarily has another endpoint. It implies that the algorithm reaches at some moment a basis $B$ without twin indices. Such a basis is such that $o\notin B$ (Section~\ref{subsec:comp}), which implies that we have a solution of the program~\eqref{pb:AMNEP} with $\omega=0$, i.e. a solution of the program~\eqref{pb:MNEP}, and thus a solution of our initial problem.


\begin{algorithm}
\begin{algorithmic}
\State{\bf Input. }The matrix $\overline{M}^{\ee}$, the matrix $M$, the vectors $\boldsymbol{b}$ and $\boldsymbol{\beta}$, an initial feasible complementary basis $B^{ini}$\;
\State{\bf Output. }A feasible basis $B^{end}$ with $o\notin B^{end}$\;
\State$\phi^{\mu}(a_0,k_0)\leftarrow$ twin index in $\mathcal{M}$\;
\State$i\leftarrow\mbox{\texttt{EnteringIndex}}(B^{ini},\phi^{\mu}(a_0,k_0))$\;
\State$j\leftarrow\mbox{\texttt{LeavingIndex}}(B^{ini},i)$\;
\State$B^{curr}\leftarrow B^{ini}\cup\{i\}\setminus\{j\}$\;
\While{There are twin indices}
\State$i\leftarrow\mbox{\texttt{EnteringIndex}}(B^{curr},j)$\;
\State$j\leftarrow\mbox{\texttt{LeavingIndex}}(B^{curr},i)$\;
\State$B^{curr}\leftarrow B^{curr}\cup\{i\}\setminus\{j\}$\;
\EndWhile
\State $B^{end}\leftarrow B^{curr}$\;
\State \Return $B^{end}$\;
\end{algorithmic}
\caption{Lemke-like algorithm} \label{algo:pivot_compl}
\end{algorithm}

\subsection{Algorithm and main result} \label{subsec:lemke}

We are now in a position to describe the full algorithm under the non-degeneracy assumption.

\begin{itemize}
\item For each $k\in K$, compute a collection $\mathcal{T}=(T^k)$ where $T^k\subseteq A^k$ is an $s^k$-arborescence of $(V^k,A^k)$.
\item Define $\ee$ as in Equation~\eqref{eq:defe} (which depends on $\mathcal{T}$).
\item Define $Y=\{\phi^x(a,k):\,a\in T^k, k\in K\}\cup\{\phi^{\mu}(a,k):\,a\in A^k\setminus T^k, k\in K\}\cup\{o\}$.
\item If $Y\setminus \{o\}$ is a complementary feasible basis providing an optimal solution of the program~\eqref{pb:AMNEP} with $\omega=0$, then we have a solution of the program~\eqref{pb:MNEP}, see Lemma~\ref{lem:initbasis}.
\item Otherwise, let $B^{ini}$ be defined as in Lemma~\ref{lem:initbasis} and apply Algorithm~\ref{algo:pivot_compl}, which returns a basis $B^{end}$.
\item Compute the basic solution associated to $B^{end}$.
\end{itemize}

All the elements proved in Section~\ref{subsec:tools} lead to the following result.

\begin{theorem}\label{thm:main}
Under the non-degeneracy assumption, this algorithm solves the program~\eqref{pb:MNEP}.
\end{theorem}

This result provides actually a constructive proof of the existence of an equilibrium for the Multiclass Network Equilibrium Problem when the cost are affine and strictly increasing, even if the non-degeneracy assumption is not satisfied. If we compute $\boldsymbol{b}=(b_v^k)$ strictly according to the model, we have \begin{equation}\label{eq:def_b} b_v^k = \left\{\begin{array}{ll}\lambda(I^k) & \mbox{if $v=s^k$} \\ -\lambda(I^k) & \mbox{if $v=t^k$} \\ 0 & \mbox{otherwise}.\end{array}\right.\end{equation} In this case, the non-degeneracy assumption is not satisfied as it has been noted at the end of Section~\ref{subsec:init} (Remark~\ref{rem:degeneracy}).
Anyway, we can slightly perturb $\boldsymbol{b}$ and $-\boldsymbol{\beta}$ in such a way that any feasible complementary basis of the perturbated problem is still a feasible complementary basis for the original problem. Such a perturbation exists by standard arguments, see~\cite{CPS92}. Theorem~\ref{thm:main} ensures then the termination of the algorithm on a feasible complementary basis $B$ whose basic solution is such that $\omega=0$. Therefore, the algorithm solves the Multiclass Network Equilibrium Problem with affine costs in any case.\\

A consequence of Theorem~\ref{thm:main} is the following. Consider the Multiclass Network Equilibrium Problem with affine costs. If the demands $\lambda(I^k)$ and the coefficients involved in the cosrs are rational numbers, then there exists an equilibrium inducing rational flows on each arc and for each class $k$. It is reminiscent of a similar result for two-player matrix games: if the matrices involve only rational entries, there is an equilibrium involving only rational numbers \citep{Na51}.
\subsection{Computational experiments}\label{subsec:experiments}

\subsubsection{Instances}

The experiments are made on $n \times n$ grid graphs (Manhattan instances). For each pair of adjacent vertices $u$ and $v$, both arcs $(u,v)$ and $(v,u)$ are present. We built several instances on these graphs with various sizes $n$, various numbers of classes, and various cost parameters $\alpha_a^k,\beta_a^k$. The cost parameters were chosen uniformly at random such that for all $a$ and all $k$ $$\alpha_a^k\in[1,10]\quad\mbox{and}\quad\beta_a^k\in[0,100].$$ 

\subsubsection{Results}

The algorithm has been coded in C++ and tested on a PC Intel{\small\textsuperscript{\textregistered}} Core{\small\textsuperscript{\texttrademark}} i5-2520M clocked at 2.5 GHz, with 4 GB RAM. The computational results are given in Table~\ref{tab:result}. Each row of the table contains average figures obtained on five instances on the same graph and with the same number classes, but with various origins, destinations, and costs parameters. 
\begin{table}
\begin{center}
\begin{tabular}{cc|cc|ccc}
   Classes & Grid & Vertices & Arcs  & Pivots  & Algorithm~\ref{algo:pivot_compl} & Inversion \\
   & & & & &  (seconds) &  (seconds)  \\ \hline
   
   2 & 2 $\times$ 2 & 4 & 8   & 2 & $<$0.01 & $<$0.01 \\
     & 4 $\times$ 4 & 16 & 48  & 21 & 0.01 & 0.03 \\
     & 6 $\times$ 6 & 36 & 120  & 54 & 0.08 & 0.5\\
     & 8 $\times$ 8 & 64  & 224  & 129 & 0.9 & 4.0  \\ \hline
     
   3 & 2 $\times$ 2 & 4 & 8   & 4 & $<$0.01 & $<$0.01  \\
     & 4 $\times$ 4 & 16 & 48   & 33 & 0.03 & 0.1 \\
     & 6 $\times$ 6 & 36 & 120  & 97 & 0.4 & 1.9\\
     & 8 $\times$ 8 & 64  & 224  & 183 & 2.6 & 12 \\ \hline
     
   4 & 2 $\times$ 2 & 4 & 8   & 3 & $<$0.01 & $<$0.01 \\
     & 4 $\times$ 4 & 16 & 48  & 41  & 0.06 & 0.3 \\
     & 6 $\times$ 6 & 36 & 120  & 126  & 0.9 & 4.7 \\
     & 8 $\times$ 8 & 64  & 224 & 249 & 5.4 & 25 \\ \hline
     
   10 & 2 $\times$ 2 & 4 & 8   & 11 & $<$0.01 & 0.02 \\
     & 4 $\times$ 4 & 16 & 48 & 107 & 0.7 & 4.1  \\
     & 6 $\times$ 6 & 36 & 120  & 322 & 15 & 70 \\
     & 8 $\times$ 8 & 64  & 224  & 638  & 87 & 385  \\ \hline
     
   50 & 2 $\times$ 2 & 4 & 8  &  56 & 0.3 & 2.6 \\
     & 4 $\times$ 4  & 16 & 48 & 636 & 105 & 511  \\

\end{tabular}
\end{center}
\caption{Performances of the complete algorithm for various instance sizes}
\label{tab:result}
\end{table}

The columns ``Classes'', ``Vertices'', and ``Arcs'' contain respectively the number of classes, the number of vertices, and the number of arcs. The column ``Pivots'' contains the number of pivots performed by the algorithm. They are done during Step 5 in the description of the algorithm in Section~\ref{subsec:lemke} (application of Algorithm~\ref{algo:pivot_compl}). The column ``Algorithm~\ref{algo:pivot_compl}'' provides the time needed for the whole execution of this pivoting step. The preparation of this pivoting step requires a first matrix inversion, and the final computation of the solution requires such an inversion as well. The times needed to perform these inversions are given in the column ``Inversion''. The total time needed by the complete algorithm to solve the problem is the sum of the ``Algorithm~\ref{algo:pivot_compl}'' time and twice the ``Inversion'' time, the other steps of the algorithm taking a negligible time.

It seems that the number of pivots remains always reasonable. Even if the time needed to solve large instances is sometimes important with respect to the size of the graph, the essential computation time is spent on the two matrix inversions. The program has not been optimized, since there are several efficient techniques known for inverting matrices. The results can be considered as very positive.

\bibliographystyle{plainnat}
\bibliography{CongestionGames} 

\end{document}